\documentclass[lettersize,journal]{IEEEtran}
\usepackage{amsmath,amsfonts}
\usepackage{algorithmic}
\usepackage{array}
\usepackage[caption=false,font=normalsize,labelfont=sf,textfont=sf]{subfig}
\usepackage{svg}
\usepackage{textcomp}
\usepackage{booktabs}
\usepackage{stfloats}
\usepackage{url}
\usepackage{verbatim}
\usepackage{algorithm}

\usepackage{graphicx}
\usepackage{cite}
\usepackage[x11names]{xcolor}
\usepackage{csquotes}
\usepackage{xcolor}  
\usepackage{pifont}

\usepackage{hyperref}
\usepackage[all]{hypcap} 
\usepackage{blindtext}
\usepackage{multirow}
\usepackage{algorithm} 
\usepackage{amsmath,amssymb}
\usepackage{amsthm}
\usepackage{booktabs}
\usepackage{diagbox}
\theoremstyle{definition} 
\newtheorem{definition}{Definition}
\newtheorem{theorem}{Theorem}
\renewcommand{\eqref}[1]{\textup{\hyperref[#1]{(\ref*{#1})}}}
\hypersetup{
  colorlinks=true,
  linkcolor=blue,   
  citecolor=blue,   
}

\begin{document}
\pagenumbering{gobble}
\title{Lipschitz-Guided Machine Learning Framework for Accelerating Transient Stability Analysis of Flexible Traction Power Supply Systems}

\author{%
  Zhong~Liu,~\IEEEmembership{Student Member,~IEEE},~%
  Jialin~Zheng,~\IEEEmembership{Member,~IEEE},~%
  Xiaonan~Lu,~\IEEEmembership{Member,~IEEE}
\thanks{(Corresponding author: Xiaonan Lu.)}%
\thanks{Zhong Liu and Xiaonan Lu are with Purdue University, 401 N. Grant St., West Lafayette, IN 47907 USA (e-mail: liu3746@purdue.edu, lu998@purdue.edu, tel.: +1-765-494-0303).}%
\thanks{Jialin Zheng is with Princeton University, Princeton, NJ 08544 USA (e-mail: jz8197@princeton.edu).}%

\thanks{This work has not been presented at a conference.}}

\markboth{IEEE}%
{Liu \MakeLowercase{\textit{et al.}}: Lipschitz-Guided Transient Stability Analysis of Flexible Traction Power Supply Systems}


\maketitle

\begin{abstract}
The rapidly varying and moving train loads in electrified railways impose stringent requirements on the transient stability of flexible traction power supply systems (FTPSSs). However, the complex dynamic interactions among numerous traction substations (TSs) make it challenging for existing transient stability analysis (TSA) methods to achieve both high accuracy and computational efficiency. To bridge this gap, this paper proposes a Lipschitz-guided machine learning framework that leverages Lipschitz continuity to restructure the transient stability verification of FTPSSs. By replacing computationally intensive satisfiability modulo theories (SMT) solvers with a deterministic algebraic check on a finite mesh, the proposed method ensures the Lyapunov conditions between mesh points without simplifications of the FTPSS dynamics. Validated on FTPSSs with three to five TSs through extensive case studies, the proposed framework accelerates training by up to 12 times compared with SMT-based neural Lyapunov methods and certifies the five-TS system. Furthermore, the proposed framework achieves up to a 33.5-fold increase in the estimated region of attraction (ROA) compared with linear matrix inequality.This advancement provides an efficient TSA solution for droop-controlled FTPSSs.
\end{abstract}


\section{Introduction}
\IEEEPARstart{T}{he} rapid expansion of electrified railways demands traction power supply systems (TPSSs) capable of delivering reliable and continuous power to support high-density rail operations~\cite{hu2024traction}. In conventional 25-kV ac TPSSs, adjacent feeding sections are connected to different phases of the utility grid to mitigate three-phase load imbalance, requiring electrical isolation through neutral sections~\cite{hu2024traction,xing2023thermal}. As trains traverse these unpowered neutral sections, they experience temporary interruptions in tractive effort, accompanied by mechanical speed drops and severe switching overvoltages~\cite{zhang2023thyristor,ge2024wide}. Furthermore, this electrical segmentation blocks the exchange of regenerative braking energy between neighboring feeding sections~\cite{chen2024power,luo2025multiagent}. To eliminate neutral sections, flexible TPSSs (FTPSSs) have emerged and are increasingly deployed in both new railway lines and substation retrofits~\cite{chen2022optimal,ge2026series}. In an FTPSS, multiple converter-based traction substations (TSs) actively regulate their output voltage magnitude, frequency, and phase to feed a continuous catenary in parallel~\cite{li2025dispatch,ying2023online}. Because the catenary voltage is no longer dictated by stiff utility transformers but is jointly established by parallel power electronic converters, the dynamic behavior of the entire traction network is fundamentally governed by the control of the TSs.

To establish a reliable catenary voltage and share traction loads, the parallel TSs typically operate as grid-forming (GFM) converters under droop control~\cite{li2022adrc,li2023spc,guo2025multitimescale}. By locally adjusting the output voltage phase angle and magnitude in response to active and reactive power variations, droop-controlled TSs autonomously distribute the traction load and seamlessly pick up the demand if an adjacent TS trips~\cite{hu2024decentralized,wang2024emergency}. At the same time, these droop-controlled TSs are electrically interconnected through the catenary impedance, which couples their power-angle dynamics and requires them to maintain synchronism to preserve system stability. As the dynamic interactions within modern electrified railway systems become increasingly complex, investigating and ensuring system stability is of paramount importance.
\begin{figure}
    \centering
    \includegraphics[width=0.85\linewidth]{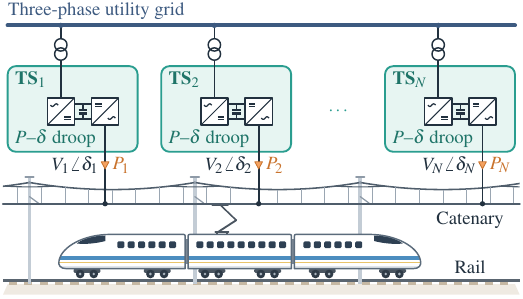}
    \caption{Illustration of flexible traction power supply systems.}
    \label{1}
\end{figure}

Stability issues in electrified railway traction power systems have been extensively investigated in the existing literature~\cite{meng2025neural,liu2024impedance,chen2026frequency}.  For FTPSSs, although the TSs operate as converter-based voltage sources with distinct synchronization dynamics, existing stability analyses remain confined to linearized models, employing pole analysis for a single static power converter~\cite{li2022adrc} and eigenvalue analysis for multiple droop-controlled TSs~\cite{li2023spc}. Because these linearized methods hold only around a steady operating point, they cannot determine whether parallel TSs can resynchronize following a large disturbance~\cite{intro}. Such large disturbances pose a severe threat to FTPSSs. On the continuous catenary, a short circuit, TS outage, or system reconfiguration~\cite{wang2024emergency,chen2026flexible} subjects all parallel TSs to large phase-angle excursions. Consequently, the TSs may fail to resynchronize even when the post-disturbance operating point is small-signal stable~\cite{TDS,intro}, generating severe circulating currents~\cite{li2025dispatch} that can trip protective relays and interrupt the traction supply. However, these fault-induced scenarios have been investigated solely through time-domain simulations and real-time tests~\cite{guo2025multitimescale,chen2026flexible}, which evaluate only individual trajectories and cannot quantify the stability margin, i.e., the region
of attraction (ROA). Post-fault transient behavior thus remains an open issue for FTPSSs~\cite{guo2024bipolar}, and no ROA of droop-controlled TSs has yet been certified.

Although Lyapunov's direct method~\cite{nonlinear} certifies ROAs in other power electronic applications, such as dc shipboard microgrids~\cite{zheng2025shipboard} and dual-active-bridge converters~\cite{xu2025dab}, directly extending existing ROA tools to ac FTPSSs faces severe limitations. Since the lossy traction catenary and distributed train loads introduce substantial transfer conductances among the TSs, no general analytical energy function exists~\cite{chiang1989existence}. Classical optimization-based constructions, notably linear matrix inequalities (LMIs)~\cite{LMI} and sum-of-squares (SOS) programming~\cite{SOS}, must either bound the trigonometric power flow conservatively or approximate the dynamics with polynomials, yielding overly conservative or inexact ROAs. Recently, neural Lyapunov functions have avoided these simplifications by learning directly on the exact nonlinear dynamics~\cite{Tong,ECCE,SCU,my}, but their formal verification relies on satisfiability modulo theories (SMT) solvers~\cite{sicun} whose computation time grows rapidly with the state dimension~\cite{chuchu}. This computational bottleneck renders SMT-based methods impractical for FTPSSs: as the traction line scales to four or five parallel TSs, SMT verification takes hours or fails to converge even for a single operating condition. Furthermore, because high-speed train movement and fluctuating traction loads continually alter the operating configuration, certifying stability margins across diverse railway operating scenarios requires a fast, SMT-free verification approach.

To overcome these limitations, this paper proposes a Lipschitz-guided neural Lyapunov framework that evaluates the transient stability of droop-controlled FTPSSs without SMT solvers. Instead of searching the continuous state space for counterexamples like SMT-based methods, the proposed approach evaluates the Lyapunov condition on a finite mesh and leverages Lipschitz constants of both the vector field and the neural Lyapunov candidate to extend this discrete check to the continuous region. The main contributions of this paper are summarized as follows.
\begin{enumerate}
    \renewcommand{\labelenumi}{\arabic{enumi})}
    \item A Lipschitz-guided neural Lyapunov learning framework for large-signal ROA certification in FTPSSs. By deriving Lipschitz constants for the FTPSS dynamics and the neural candidate, the proposed method transforms formal verification from an expensive continuous-space SMT search into a deterministic algebraic check on a finite mesh.
    \item The deterministic verification mechanism substantially accelerates the training process and scales neural Lyapunov analysis to multi-TS traction networks that are computationally intractable for SMT-based methods.
    \item The computational acceleration is achieved without compromising ROA estimation quality. Operating directly on the nonlinear angle-droop dynamics, the proposed framework maintains stability guarantees and estimates less conservative ROAs than existing methods across all tested cases.
\end{enumerate}

The remainder of this paper is organized as follows. Section~II establishes the nonlinear model of the droop-controlled FTPSS and formulates the ROA certification problem. Section~III develops the Lipschitz-guided verification method, and Section~IV details the training procedure of the neural Lyapunov candidate. Section~V demonstrates the framework's efficiency through extensive case studies. Finally, Section~VI concludes the paper.

\section{Preliminaries}

\subsection{Modeling and Control Strategy of FTPSSs}

As shown in Fig.~\ref{1}, $N$ TSs feed a continuous single-phase catenary in parallel, and the trains draw traction power from it through their pantographs. Each TS connects the three-phase utility grid to the catenary through a transformer and an ac--dc--ac converter. Each TS is represented by the GFM voltage source $E_k\angle\delta_k$ formed by its traction-side converter~\cite{li2022adrc}, where $\delta_k$ is measured in a common reference frame rotating at the nominal angular frequency. Under angle and voltage droop control, the $k$-th TS adjusts the voltage according to its measured active and reactive powers $P_k$ and $Q_k$:
\begin{equation}
\tau_{a}\,\dot{\delta}^{'}_k + \delta^{'}_k
= m_{a}\bigl(P^{*}_k - P_k\bigr)
\label{eq:dyn_angle}
\end{equation}
\begin{equation}
\tau_{v}\,\dot{E}^{'}_k + E^{'}_k
= n_{v}\bigl(Q^{*}_k - Q_k\bigr)
\label{eq:dyn_volt}
\end{equation}
where variables $\delta'_k$ and $E^{'}_k$ denote the deviations of the voltage phase angle and voltage magnitude from their respective nominal steady-state values, i.e., $\delta_k':= \delta_k - \delta_k^*$, $E_k' := E_k - E_k^*$. $\tau_a$ and $\tau_v$ represent the time constants of the low-pass filters; $m_a$ and $n_v$ are the angle and voltage droop gains, respectively. $P_k$ and $Q_k$ indicate the real and reactive power injections, while $P^*_k$ and $Q^*_k$ signify the steady-state references.

When $N$ TSs are interconnected through the catenary, the powers delivered by the $k$-th TS satisfy:
\begin{align}
P_{k} - G_{kk}E_{k}^{2} - \sum_{i\neq k} E_{k}E_{i}Y_{ki}\,\cos\bigl(\delta_{ki}-\varphi_{ki}\bigr) &= 0 \label{eq:pf_P} \\
Q_{k} + B_{kk}E_{k}^{2} - \sum_{i\neq k} E_{k}E_{i}Y_{ki}\,\sin\bigl(\delta_{ki}-\varphi_{ki}\bigr) &= 0 \label{eq:pf_Q}
\end{align}
In \eqref{eq:pf_P} and \eqref{eq:pf_Q}, $\delta_{ki} = \delta_{k} - \delta_{i}$. The terms $G_{kk}$, $B_{kk}$, $Y_{ki}$, and $\varphi_{ki}$ are parameters derived from the catenary admittance matrix. The steady-state operating points are determined by setting the time derivatives in \eqref{eq:dyn_angle} and \eqref{eq:dyn_volt} to zero.

In practical operations, the dynamics of phase angles and voltage magnitudes are largely decoupled. Since voltage magnitudes can typically be regulated via local compensation, transient stability degradation is predominantly driven by the phase angles. Therefore, focusing on these dominant angle dynamics, we define the global state vector $\mathbf{x} \in \mathbb{R}^{N}$ by aggregating the state variables of all TSs:
\begin{equation}
\mathbf{x} = \left[ \delta'_1, \dots, \delta'_N \right]^\top
\label{eq:state_vector}
\end{equation}

\subsection{Transient Stability Formulation of  FTPSSs}

To evaluate the transient stability of FTPSSs following severe disturbances, Lyapunov's direct method is adopted to assess the asymptotic stability~\cite{nonlinear} of the system. Specifically, a scalar energy-like function, i.e., the Lyapunov function $V(\mathbf{x})$, is introduced. The system equilibrium is guaranteed to be asymptotically stable if $V(\mathbf{x})$ is positive definite and its time derivative along the system trajectories is strictly negative. Numerically, this condition is characterized by the Lie derivative, defined as:
\begin{equation}
    \mathcal{L}_{f} V(\mathbf{x}) = \nabla V(\mathbf{x}) \cdot f(\mathbf{x}) < 0, \qquad \forall\,\mathbf{x}\neq\mathbf{0}
    \label{lie}
\end{equation}
where $f(\mathbf{x})$ denotes the nonlinear vector field of the FTPSSs (the definition of $f(\mathbf{x})$ can be found in Appendix \ref{deflip}).

For practical FTPSSs, determining the exact stability boundary is typically intractable. Therefore, the transient stability margin is conservatively quantified by ROA, i.e., any fault-induced state deviation that remains within ROA is guaranteed to converge back to the equilibrium point, thereby verifying the post-fault transient stability of the FTPSSs.
\begin{figure*}[t]
\centering
\includegraphics[width=1\linewidth]{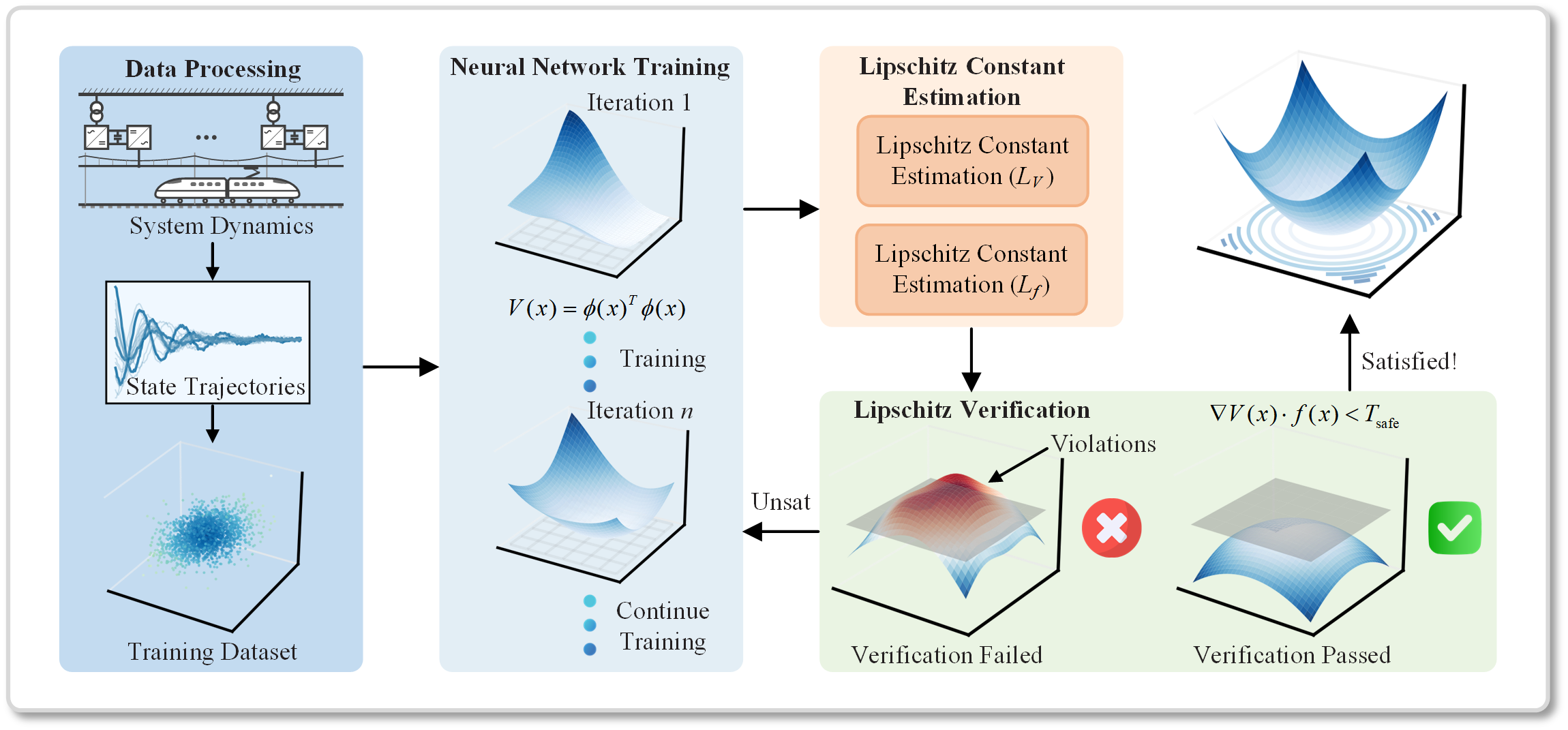}
\vspace{-0.7cm}
\caption{Schematic of the Lipschitz-guided learning framework.}
\label{framework}
\vspace{-0.7cm}
\end{figure*}
\section{The Lipschitz-guided Neural Lyapunov Learning Framework}
\label{III}
\subsection{Overall Framework}
\label{overview}
As illustrated in Fig. \ref{framework}, the proposed Lipschitz-guided learning framework can be divided into two interdependent stages: data-driven neural Lyapunov function training and Lipschitz verification. The first stage focuses on constructing a neural Lyapunov candidate $V_\theta(x)$ through empirical loss minimization, and the core contribution of this paper is subsequently realized in the second stage: a deterministic verification paradigm tailored for FTPSSs.

Historically, certifying the validity of neural Lyapunov functions relies heavily on SMT solvers. However, when applied to FTPSSs, the traditional SMT-based method inherently suffers from the curse of dimensionality. The exact verification over an infinite continuous state space is often computationally intractable, leading to scalability bottlenecks for practical railway applications.

To circumvent this limitation, we propose a paradigm shift from traditional exhaustive search to a Lipschitz-guided machine learning framework. The core mechanism relies on the mathematical property of Lipschitz continuity (see Appendix \ref{deflip}). By bounding the maximum rate of change (i.e., the Lipschitz constant) of both the FTPSS dynamics and the learned Lyapunov function, the proposed framework mathematically limits the worst-case variation between any two adjacent mesh points. This critical property successfully translates a computationally intractable verification problem into a deterministic, algebraic check on a finite set of mesh points. Consequently, the infinite continuous state space is effectively reduced to a bounded mesh evaluation, thus drastically enhancing the scalability for complex FTPSSs. The technical details of this paradigm are formulated in the following subsections.
\subsection{Derivation of Lipschitz Constant Estimation}
\label{subsec:lip_estimation}

To facilitate the stability verification, the proposed framework requires estimation of the Lipschitz constants, which essentially quantifies the maximum possible fluctuation of the function value between sampled mesh points.

First, the Lipschitz constant of the FTPSS vector field $f(x)$ (see Appendix \ref{deflip}), denoted as $L_f$, is estimated. This parameter quantifies the maximum rate at which the system dynamics change with respect to the state variations. The system Jacobian matrix $A$ is constructed by leveraging automatic differentiation:
\begin{equation}
A = \left. \frac{\partial f(\mathbf{x})}{\partial \mathbf{x}} \right|_{\mathbf{x}=\mathbf{x}^*}
\end{equation}
The Lipschitz constant $L_f$ is then approximated by the spectral norm (i.e., the maximum singular value) of this linearized state matrix:
\begin{equation}
    \hat{L}_f \approx \| A \|_2 = \sigma_{\max}(A)
    \label{eq:Lf_est}
\end{equation}

The second critical parameter involves the geometric properties of the learned neural network. Unlike the fixed system dynamics, the neural Lyapunov candidate $V_\theta(x)$ is a highly nonlinear function that evolves during training. To estimate the variation of the energy function value with respect to state perturbations, the Lipschitz constant of the Lyapunov function, denoted as $L_{V}$, is calculated empirically over the structural verification mesh $\mathcal{G}$. This process involves computing the gradient vector $\nabla V_\theta(x_g)$ for every node $x_g \in \mathcal{G}$ via backpropagation and determining the maximum Euclidean norm observed across the mesh:
\begin{equation}
\hat{L}_{V} := \max_{x_g \in \mathcal{G}} \big\| \nabla V_\theta(x_g) \big\|_2
\label{eq:Lv_est}
\end{equation}
The parameter $\hat{L}_{V}$ characterizes the maximum steepness of the learned energy landscape. By utilizing smooth activation functions (hyperbolic tangent) and regularizing the training process, the network is encouraged to learn smooth manifolds. Physically, $\hat{L}_{V}$ limits the maximum slope of the Lyapunov surface, guaranteeing that the function value does not change drastically in the continuous space between the discrete mesh points.

\begin{figure}[!t]
\centering
\includegraphics[width=\columnwidth]{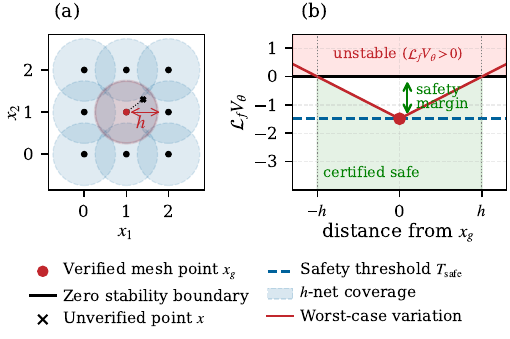}
\vspace{-0.9cm}
\caption{Illustration of Lipschitz verification: (a) Domain discretization via $h$-net coverage. (b) Construction of the safety threshold $T_{\mathrm{safe}}$ to bound the worst-case variation of $\mathcal{L}_{f} V_\theta(x)$ in inter-mesh regions.}
\label{Lipschitz}
\vspace{-0.7cm}
\end{figure}

\subsection{Stability Verification with Estimated Lipschitz Constant }
\label{subsec:verification}

While the neural network is trained to minimize Lyapunov loss on finite discrete samples, certifying asymptotic stability requires guaranteeing $\mathcal{L}_{f} V_\theta(x) < 0$ continuously over the entire domain $\Omega$. To bridge the gap between discrete evaluation and continuous certification, a verification procedure is executed utilizing the Lipschitz constant estimated in Section~\ref{subsec:lip_estimation}.

The verification process begins by discretizing the continuous domain $\Omega$ into a structural mesh $\mathcal{G}$. A critical parameter for this discretization is the \textit{spatial resolution} $h$, as visualized in the global mesh coverage illustration of Fig.~\ref{Lipschitz} (a). Formally, $h$ is defined as the maximum Euclidean distance from any point in the continuum $\Omega$ to its nearest mesh point in $\mathcal{G}$:
\begin{equation}
    h = \sup_{x \in \Omega} \min_{x_g \in \mathcal{G}} \|x - x_g\|_2
\end{equation}
This definition ensures that the mesh $\mathcal{G}$ constitutes a $h$-net, where the union of the $h$-circles centered at each mesh point forms a complete covering of the region of interest, leaving no state unmonitored.

To ensure stability between these discrete mesh points, a safety margin strategy is employed to counteract potential inter-mesh discretization errors. Different from traditional verification conditions that merely verify that the Lie derivative is negative ($\mathcal{L}_{f} V_\theta(x) < 0$) at sampled mesh points, we adopt a tightened verification condition by enforcing a carefully designed threshold to prevent function value at intermediate points potentially deviating and breaching the $\mathcal{L}_{f} V_\theta(x) =0$ plane. By constraining the derivative at each mesh point to remain below this threshold, the intrinsic property of Lipschitz continuity is leveraged to limit the maximum possible variation of the function. The proposed mechanism ensures that the Lie derivative across the entire continuum remains confined below the zero plane, thereby preventing any violation of the stability criteria in the regions between mesh points.

To visualize the verification process intuitively (as shown in the 2D illustration in Fig.~\ref{Lipschitz}\,(b)), consider the worst-case scenario where the composite Lie derivative function $\mathcal{L}_{f} V_\theta(x)$ rises at its maximum possible rate, approximated by the composite Lipschitz constant $\mathcal{K}_{comp} = L_V(1+L_f)$, as it moves away from a sampled mesh point. Over the maximum possible distance $h$, the function value can increase by at most $\mathcal{K}_{comp} \cdot h$, which constitutes the required safety margin. Therefore, to guarantee that the curve never intersects the unstable $\mathcal{L}_{f} V_\theta(x) =0$ plane, the value at the mesh point $x_g$ must start "deep enough" (i.e., sufficiently negative) to fully absorb this potential rise.

This leads to the verification criterion:
\begin{equation}
    \mathcal{L}_{f} V_\theta(x_g) < -L_V(1+L_f) h, \quad \forall x_g \in \mathcal{G}
    \label{eq:verify_condition}
\end{equation}
The inequality imposes a strict safety threshold at $T_{\mathrm{safe}}=-L_V(1+L_f) h$. If the Lie derivative at a mesh point lies below this threshold, mathematically, the function value within the surrounding $h$-neighborhood is expected to remain negative.

The verification is implemented as a global Boolean check across the finite set $\mathcal{G}$. This process can be visualized through the energy landscape surface as shown in the Lipschitz Verification block of Fig. \ref{framework}. A successful certification implies that the evaluated Lie derivative surface lies entirely beneath the safety threshold plane defined by $T_{\mathrm{safe}}$. Conversely, a verification failure occurs if any portion of the surface violates this threshold, indicated by the surface protruding above the safety plane, which signals a potential risk of positive derivatives in the inter-mesh regions. The mathematical proof, which certifies that this discrete Boolean check guarantees continuous-time asymptotic stability, is provided in Appendix~\ref{app:proof}.

\section{Setup and Implementation}
\subsection{Data Acquisition and Set Construction}
The successful implementation of the proposed verification framework requires a trained neural network to serve as the Lyapunov candidate. This section elaborates on this data-driven construction stage, beginning with data acquisition. The mathematical formulation established in Appendix \ref{deflip} constitutes the model governing the FTPSSs evolution. This model serves as the ground truth for generating the trajectory data required for the proposed learning framework.

In the \emph{Data Processing} block of Fig.~\ref{framework}, time-domain trajectories are collected to represent the system behaviors within the state space. The FTPSS is subjected to a diverse set of disturbances, including fault scenarios of varying severities, which perturb the system states to a wide range of initial positions $\mathbf{x}_0$ surrounding the equilibrium. For each event, the continuous-time evolution is discretized and recorded at a uniform sampling interval $\Delta t$, resulting in a trajectory dataset:
\begin{equation}
\mathcal{D}_{\mathrm{traj}} := \{\mathbf{x}(t_j)\}_{j=0}^{N_T}, \qquad t_j = t_0 + j\Delta t
\label{eq:traj_dataset}
\end{equation}

To effectively train the neural networks and perform stability verification, a compact region of interest (ROI), denoted as $\Omega \subset \mathbb{R}^{N}$, is defined. This region encapsulates the physically meaningful operating range of the FTPSSs states (e.g., angle and voltage limits). A comprehensive point cloud $\mathcal{X}$ is then constructed by aggregating samples derived from $\mathcal{D}_{\mathrm{traj}}$ that fall within $\Omega$:
\begin{equation}
\mathcal{X} := \{\mathbf{x}_i\}_{i=1}^{N_s} \subset \Omega
\label{eq:state_samples}
\end{equation}
This dataset $\mathcal{X}$ is utilized principally for the gradient-based training of the neural Lyapunov candidate.
\begin{figure}[t]
    \centering
    \includegraphics[width=1\linewidth]{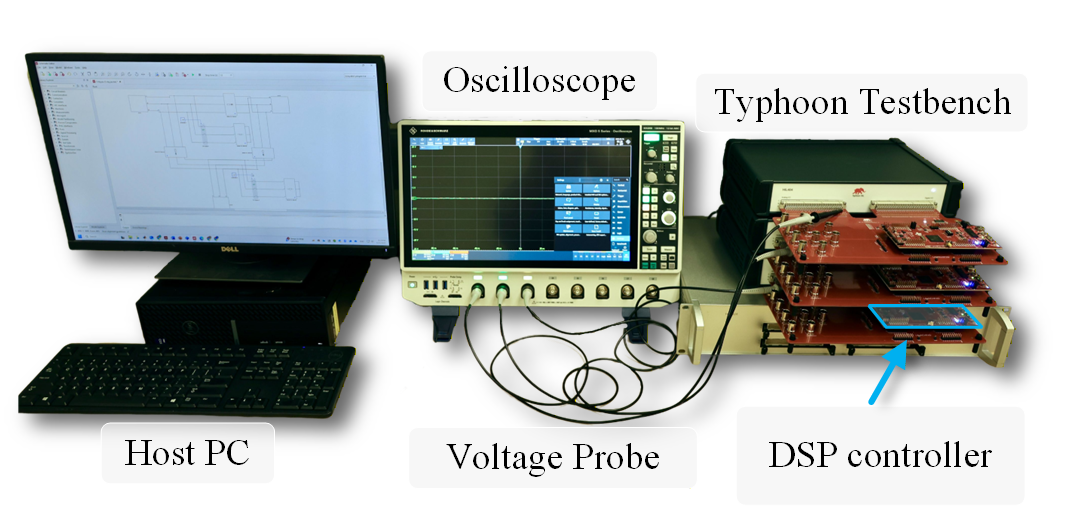}
    \caption{Versatile and scalable FTPSS testbed.}
    \label{fig:hil_setup}
\end{figure}

\subsection{Neural Network Training}
\label{subsec:nn_training}

To evaluate the transient stability of the FTPSSs via the proposed framework, a neural network is employed to construct the candidate Lyapunov function. The training process is formulated as a stability-guided learning objective, aimed at shaping the energy landscape to satisfy the Lyapunov criteria, i.e., the positive definiteness and negative definiteness of its Lie derivative \eqref{lie}.

Standard feedforward neural networks (FNN) do not inherently guarantee the mathematical properties required for stability certification. To enforce positive semidefiniteness by construction, the candidate Lyapunov function $V_\theta(x)$ is parameterized as a quadratic form of the network output, as illustrated in the \textit{Neural Network Training} block of Fig.~\ref{framework}. Let $\phi_\theta: \mathbb{R}^{N} \to \mathbb{R}^{n_\phi}$ denote a mapping parameterized by a fully connected FNN with $K$ hidden layers and hyperbolic tangent ($\tanh$) activation functions. Here, $N$ corresponds to the dimension of the system state space $\mathbf{x}$, and $n_\phi$ represents the dimension of the network's output feature vector. The candidate function is defined as:
\begin{equation}
    V_\theta(x) = \phi_\theta(x)^\top \phi_\theta(x)
    \label{eq:V_structure}
\end{equation}
This structural design guarantees that $V_\theta(x) \ge 0$ for all states $x$, thereby satisfying the positive semidefinite condition intrinsically.

The core learning objective is to shape the Lie derivative $\mathcal{L}_{f} V_\theta(x)$ to satisfy asymptotic stability conditions. To avoid trivial mathematical solutions where the candidate Lyapunov function uniformly converges to zero, a strictly positive stability margin $\rho$ is explicitly introduced. By requiring $\mathcal{L}_{f} V_\theta(x) \le -\rho$ across the state space, the framework guarantees a valid and robust transient stability certification.

The training objective is defined by a composite loss function $\mathcal{L}(\theta)$ computed over the dataset $\mathcal{X}$. The primary component penalizes violations of the strictly negative derivative condition:
\begin{equation}
    \mathcal{L}_{\mathrm{deriv}}(\theta) = \frac{1}{|\mathcal{X}|} \sum_{x_i \in \mathcal{X}} \mathrm{ReLU}\left( \mathcal{L}_fV_\theta(x_i) + \rho \right)
    \label{eq:loss_deriv}
\end{equation}
The $\mathrm{ReLU}(\cdot)$ operator ensures that a penalty is incurred only when the derivative fails to be sufficiently negative (i.e., $\mathcal{L}_fV_\theta(x_i) > -\rho$). To simultaneously ensure that the Lyapunov function vanishes at the equilibrium, a soft constraint term is incorporated:
\begin{equation}
    \mathcal{L}_{\mathrm{origin}}(\theta) = \left( V_\theta(\mathbf{0}) \right)^2
    \label{eq:loss_origin}
\end{equation}
where $\mathbf{0}$ represents the equilibrium state vector. The total optimization objective is a weighted sum of these terms:
\begin{equation}
    \min_{\theta} \quad \mathcal{L}(\theta) = \mathcal{L}_{\mathrm{deriv}}(\theta) + \lambda \mathcal{L}_{\mathrm{origin}}(\theta)
    \label{eq:total_loss}
\end{equation}
where $\lambda$ is a weighting hyperparameter. The network parameters $\theta$ are updated iteratively using the Adam optimizer. By minimizing \eqref{eq:total_loss}, the neural network is driven to approximate a valid Lyapunov function that satisfies stability conditions. Once this empirical training phase converges, the resulting network $V_\theta(x)$ is passed to the Lipschitz verification stage (Section~\ref{III}) to certify the stability. 
\section{Case Study}
To validate the proposed Lipschitz-guided learning framework, this section presents case studies on FTPSSs of increasing complexity with 3, 4, and 5 TSs and different system topologies, thereby substantiating its effectiveness. Particularly, in order to verify the efficacy of the proposed TSA approach and safely emulate various fault scenarios, a versatile testbed enabled by Typhoon Hardware-in-the-Loop (HIL) platform is implemented, as shown in Fig. \ref{fig:hil_setup}. The platform is an emulation of the FTPSS: the power stages of the TSs and the catenary network are emulated in real time on Typhoon testbench. The control algorithms for each TS are implemented on separate Digital Signal Processor (DSP) boards (TI TMS320F28379D). These controllers communicate with the HIL emulator via analog and digital I/O channels, closing the loop to replicate control interactions. A host PC and an oscilloscope are utilized to monitor system states, capture transient waveforms, and validate the theoretical stability predictions (ROA) against time-domain measurements.

\begin{figure}[t]
\centering
\includegraphics[width=1\linewidth]{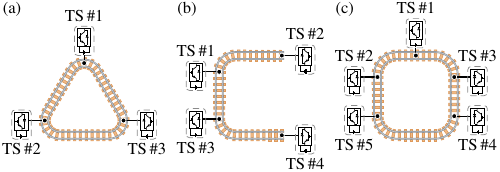}
\vspace{-0.8cm}
\caption{(a) Schematic of 3-TS FTPSS. (b) Schematic of 4-TS FTPSS. (c) Schematic of 5-TS FTPSS.}
\label{schematic}
\vspace{-0.2cm}
\end{figure}
\begin{figure}
    \centering
    \includegraphics[width=1\linewidth]{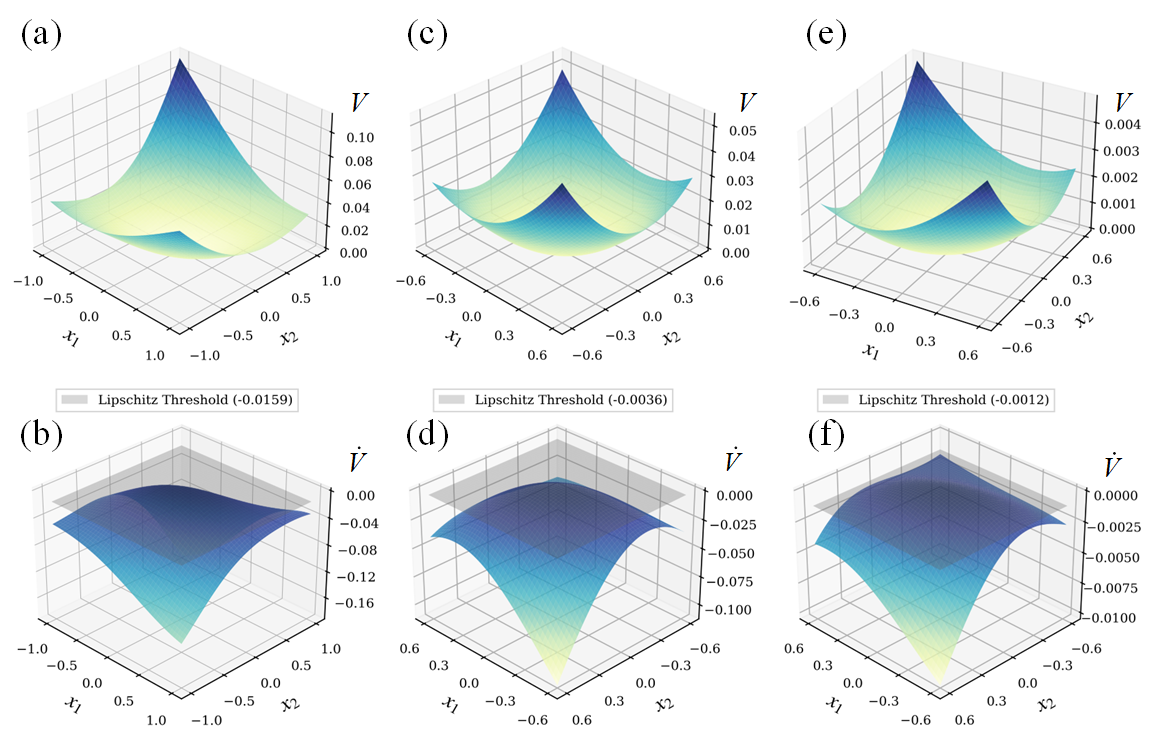}
    \caption{Training results: (a) Lyapunov function of 3-TS system. (b) Lie derivative $\mathcal{L}_{f} V_\theta(x)$ of 3-TS system. (c) Lyapunov function of 4-TS system. (d) Lie derivative $\mathcal{L}_{f} V_\theta(x)$ of 4-TS system. (e) Lyapunov function of 5-TS system. (f) Lie derivative $\mathcal{L}_{f} V_\theta(x)$ of 5-TS system.}
    \label{combined Lyas}
\end{figure}

For the 3-TS system (Case I), the system configuration is presented in Fig. \ref{schematic}\,(a), with detailed parameters provided in Table \ref{tab:3mg}. Trajectory data collected from the HIL setup were utilized to generate the training dataset, which was subsequently fed into the neural Lyapunov function training procedure. Following the Lipschitz constant estimation and verification steps, the proposed framework successfully derived a valid neural Lyapunov function and its Lie derivative within 22 seconds, as illustrated in Fig. \ref{combined Lyas}. In Fig. \ref{combined Lyas}\,(a), the learned Lyapunov function $V_\theta(x)$ exhibits a distinct convex geometry (positive definiteness), while the corresponding Lie derivative surface $\mathcal{L}_{f} V_\theta(x)$ (Fig. \ref{combined Lyas}\,(b)) remains strictly negative within the ROI except at the origin. Crucially, the Lie derivative surface is maintained below the designated Lipschitz safety margin. This geometric property is pivotal: it not only ensures negative definiteness on the discrete sampled points but also ensures negative definiteness over the continuous state space of interest, thereby certifying the system's asymptotic stability. A quantitative comparison with existing methods is summarized in Table \ref{case_study_comparison}. In terms of computational efficiency, the proposed method is approximately three times faster than the neural Lyapunov method and six times faster than the LMI approach. Regarding the SOS method, although it exhibits a shorter computation time, it is subject to a fundamental limitation when applied to FTPSSs. Since SOS optimization requires the system dynamics to be polynomial, the intrinsic sinusoidal nonlinearities of the FTPSS have to be approximated via Taylor expansion prior to optimization. Consequently, the stability certificate obtained by SOS is valid only for the local approximated model, rather than the actual high-order nonlinear system. The approximation inevitably introduces errors that compromise the reliability of the stability assessment. In contrast, the proposed method directly processes the full nonlinear dynamics without requiring any model simplification or polynomial approximation. This ensures that the synthesized Lyapunov function faithfully captures the true system behavior. Furthermore, the ROA estimated by the proposed method is less conservative than that of the counterparts, as visualized in Fig. \ref{ROA}. A detailed quantitative performance metric is provided in Fig. \ref{metric}.
\begin{table}[htbp]
    \centering
    
    \caption{System Parameters for the 3-TS Case}
    \label{tab:3mg}
    \resizebox{\columnwidth}{!}{%
    \begin{tabular}{@{}ccccc ccc@{}}
    \toprule
    \multicolumn{5}{c}{\textbf{TS Parameters}} & \multicolumn{3}{c}{\textbf{Section Parameters}} \\
    \cmidrule(r){1-5} \cmidrule(l){6-8}
    \textbf{TS} & $\tau_{ai}$ & $m_{ai}$ & $P_i^*$ & $\delta_i^*$ & \textbf{Section} & $G_{ij}$ & $B_{ij}$ \\
    \textbf{Index} & (p.u.) & (p.u.) & (p.u.) & (rad) & $i-j$ & (p.u.) & (p.u.) \\
    \midrule
    1 & 1.2 & 0.2 & 0.1995 &  0.0000  & 1-2 & 0.4528 & -0.4151 \\
    2 & 1.0 & 0.2 & 1.3644 &  0.9716  & 1-3 & 0.3659 & -0.3317 \\
    3 & 0.8 & 0.2 & 0.0476 & -0.7919  & 2-3 & 0.4192 & -0.3308 \\
    \midrule 
    \multicolumn{8}{c}{System Base Values: $S_{\text{base}} = 1$ MVA, $V_{\text{base}} = 480$ V} \\
    \bottomrule
    \end{tabular}%
    }
    
    \vspace{1.5em} 
    
    \caption{System Parameters for the 4-TS Case}
    \label{tab:4mg}
    \resizebox{\columnwidth}{!}{%
    \begin{tabular}{@{}ccccc ccc@{}}
    \toprule
    \multicolumn{5}{c}{\textbf{TS Parameters}} & \multicolumn{3}{c}{\textbf{Section Parameters}} \\
    \cmidrule(r){1-5} \cmidrule(l){6-8}
    \textbf{TS} & $\tau_{ai}$ & $m_{ai}$ & $P_i^*$ & $\delta_i^*$ & \textbf{Section} & $G_{ij}$ & $B_{ij}$ \\
    \textbf{Index} & (p.u.) & (p.u.) & (p.u.) & (rad) & $i-j$ & (p.u.) & (p.u.) \\
    \midrule
    1 & 1.2 & 1.2 & 0.2604  & 0.0000  & 1-2 & 0.4515 & -0.4141 \\
    2 & 1.0 & 1.2 & -0.5843 & -1.0472 & 1-3 & 0.6403 & -0.4601 \\
    3 & 0.8 & 1.2 & 1.1924  & 2.3562  & 3-4 & 0.3669 & -0.3306 \\
    4 & 1.0 & 1.2 & -0.2244 & 0.5236  & --  & --      & --      \\
    \midrule 
    \multicolumn{8}{c}{System Base Values: $S_{\text{base}} = 1$ MVA, $V_{\text{base}} = 480$ V} \\
    \bottomrule
    \end{tabular}%
    }
    
    \vspace{1.5em} 
    
    \caption{System Parameters for the 5-TS Case}
    \label{tab:5mg}
    \resizebox{\columnwidth}{!}{%
    \begin{tabular}{@{}ccccc ccc@{}}
    \toprule
    \multicolumn{5}{c}{\textbf{TS Parameters}} & \multicolumn{3}{c}{\textbf{Section Parameters}} \\
    \cmidrule(r){1-5} \cmidrule(l){6-8}
    \textbf{TS} & $\tau_{ai}$ & $m_{ai}$ & $P_i^*$ & $\delta_i^*$ & \textbf{Section} & $G_{ij}$ & $B_{ij}$ \\
    \textbf{Index} & (p.u.) & (p.u.) & (p.u.) & (rad) & $i-j$ & (p.u.) & (p.u.) \\
    \midrule
    1 & 1.2 & 1.2 & 0.2604  & 0.0000  & 1-2 & 0.4515 & -0.4141 \\
    2 & 1.0 & 1.2 & -1.0028 & -1.0472 & 1-3 & 0.6403 & -0.4601 \\
    3 & 0.8 & 1.2 & 1.1924  & 2.3562  & 2-5 & 0.3801 & -0.3427 \\
    4 & 1.0 & 1.2 & -0.5612 & 0.5236  & 3-4 & 0.3669 & -0.3306 \\
    5 & 1.2 & 1.2 & -0.2338 & 0.3000  & 4-5 & 0.4207 & -0.3310 \\
    \midrule 
    \multicolumn{8}{c}{System Base Values: $S_{\text{base}} = 1$ MVA, $V_{\text{base}} = 480$ V} \\
    \bottomrule
    \end{tabular}%
    }
\end{table}

To validate the accuracy of the estimated ROA, the 3-TS system was subjected to various large-signal disturbances for transient stability assessment. 
In the first test case, a catenary-to-rail short circuit with a fault resistance of $R_f=0.08\,\Omega$ was applied across sections 1-2, 1-3, and 2-3 during the interval $t=5.5$-$6.0$\,\text{s}. This disturbance displaced the system state trajectories to a post-fault initial condition of $[\delta_1, \delta_2, \delta_3] = [-0.18, 0.64, -0.97]\,\text{rad}$, corresponding to the yellow triangle in Fig. \ref{waves}\,(a). As shown in the time-domain waveforms, following the fault clearance, both the phase angles and the output current of TS \#1 asymptotically converged to their pre-fault equilibrium points, confirming the stability prediction within the estimated ROA. Subsequently, a longer catenary-to-rail short circuit ($R_f=0\,\Omega$) was applied to sections 1-2 and 2-3 during $t=5.0$--$6.0$~s. This fault perturbed the system states to $[\delta_1, \delta_2, \delta_3] = [0.06, 0.83, -0.76]\,\text{rad}$, marked by the red dot in Fig. \ref{waves}\,(a). As shown by the waveforms, the phase angles stabilized, and the output current of TS \#1 successfully recovered to its nominal steady-state magnitude after the fault was cleared. Collectively, these time-domain results demonstrate that the system maintains stability under distinct perturbation levels, thereby validating the effectiveness and reliability of the estimated ROA.

\begin{table*}[b]
\centering
\small
\caption{Comparison of Computation Speed and ROA Size Across Three Case Studies}
\label{case_study_comparison}
\setlength{\tabcolsep}{5pt}
\renewcommand{\arraystretch}{1.2}
\begin{tabular}{lcccccccc}
\toprule
\multirow{2}{*}{\textbf{Method}} &
\multicolumn{2}{c}{\textbf{Case I}} &
\multicolumn{2}{c}{\textbf{Case II}} &
\multicolumn{2}{c}{\textbf{Case III}} \\
\cmidrule(lr){2-3}\cmidrule(lr){4-5}\cmidrule(lr){6-7}\cmidrule(lr){8-9}
 & Time (s) & ROA (\%) & Time (s) & ROA (\%) & Time (s) & ROA (\%) \\
\midrule
Lipschitz Lyapunov          & 22    & 51.33 & 295   & 52.74 & 2361  & 62.59 \\
Neural Lyapunov             & 72    & 41.39 & 3664  & 46.27 & 48600 & \textit{Fail} \\
Sum-of-Squares                & 4     & 21.84 & 51.8  & 29.57 & 1201  & 48.12 \\
Linear Matrix Inequality     & 130   & 22.43 & 6651  & 23.79 & 10928 & 1.87 \\
\bottomrule
\end{tabular}%
\end{table*}

To evaluate the scalability of the proposed framework, the methodology is extended to higher-order FTPSSs: a 4-TS topology (Case II) and a 5-TS topology (Case III). The corresponding system topologies are illustrated in Fig. \ref{schematic}, with detailed parameters provided in Table \ref{tab:4mg} and Table \ref{tab:5mg}, respectively. The trained Lyapunov functions and Lie derivatives for these two cases are shown in Fig. \ref{combined Lyas}. As evident from the visualizations, the learned Lyapunov functions exhibit positive definiteness, while their corresponding Lie derivatives consistently remain below the Lipschitz threshold, thereby ensuring the system's asymptotic stability.
\begin{figure}[t]
\centering
\includegraphics[width=1\linewidth]{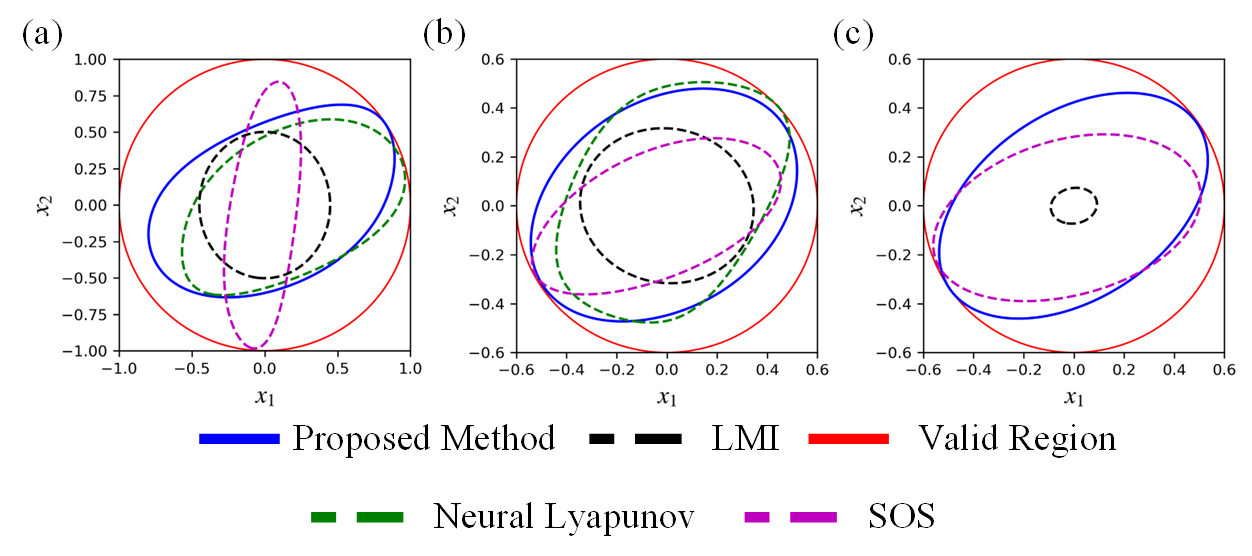}
\caption{ROA comparison by different methods across three case studies: (a) Case I (3-TS); (b) Case II (4-TS); and (c) Case III (5-TS).}
\label{ROA}
\end{figure}

A pivotal advantage of the proposed approach is the computational efficiency in high-dimensional state spaces. As illustrated in Fig. \ref{metric}\,(a) and Table \ref{case_study_comparison}, traditional analytical methods such as SOS and LMI suffer from dramatic complexity growth as the system dimension increases. Specifically, for the 5-TS case, the neural Lyapunov method required over 13 hours but failed to converge, while the SOS method required extensive preprocessing for polynomial approximation. In contrast, the proposed Lipschitz-guided framework was completed within 295s for the 4-TS system and 2361s for the 5-TS system. These results demonstrate that the computational burden of the proposed method scales moderately with system complexity, making it tractable for larger-scale networks.
\begin{figure}
    \centering
    \includegraphics[width=0.8\linewidth]{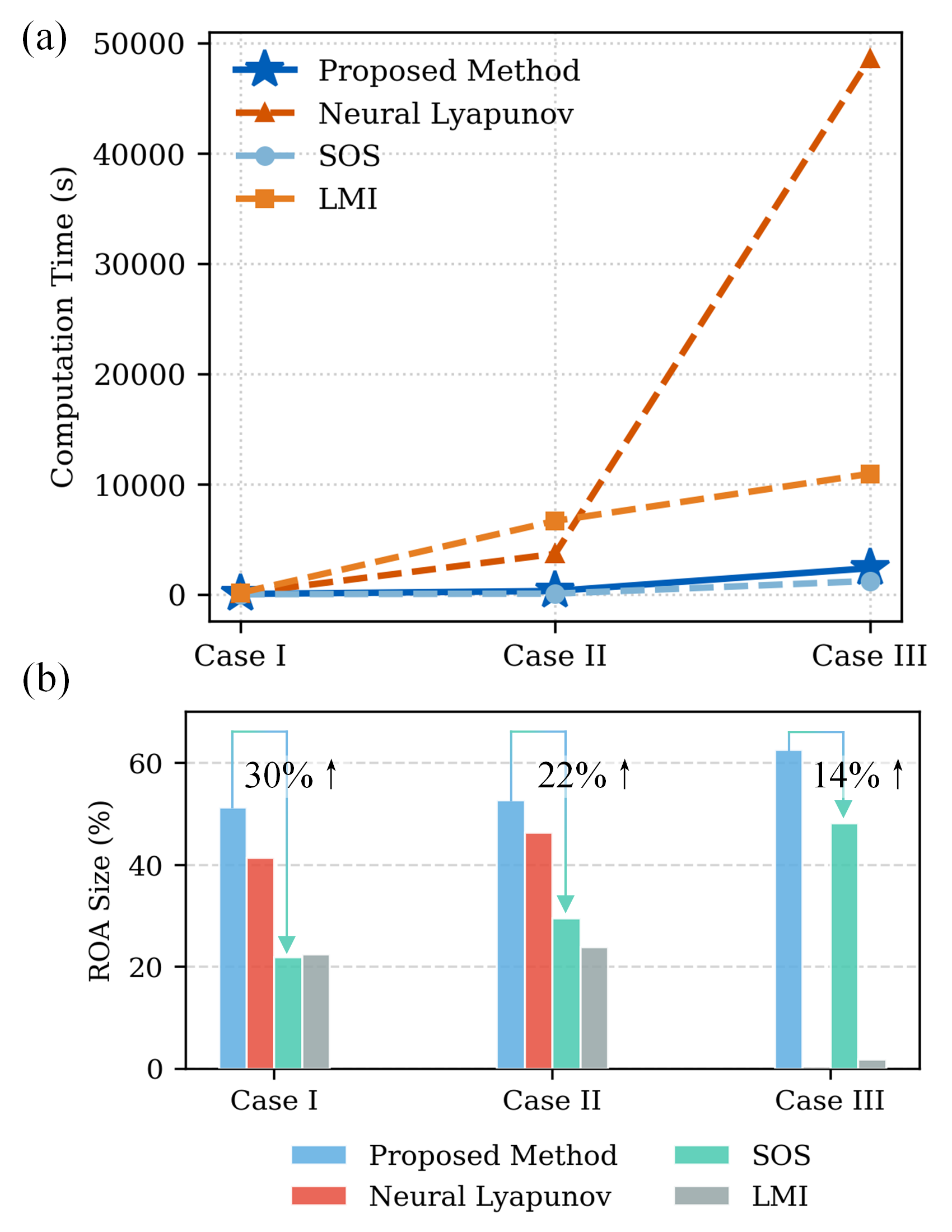}
    \caption{Quantitative performance evaluation across three case studies: (a) Computation time comparison. (b) ROA size comparison.}
    \label{metric}
\end{figure}

Regarding estimation accuracy, Table \ref{case_study_comparison} and Fig. \ref{metric}\,(b) reveal that the proposed method consistently yields the least conservative stability regions. The ROA (\%) values in Table~\ref{case_study_comparison} are estimated by Monte Carlo sampling. Notably, in the 5-TS case, the ROA estimated by the LMI method collapses to a negligible volume ($1.87\%$) due to the conservative interval-bounded approximations (e.g., min-max bounds in Takagi-Sugeno modeling). Conversely, the proposed method successfully identifies a substantial stability region ($62.59\%$), highlighting its capability to directly address the intrinsic sinusoidal nonlinearities in FTPSSs without the excessive conservatism introduced by model simplifications.
\begin{figure*}
    \centering
    \vspace{-0.3cm}
    \includegraphics[width=1\linewidth]{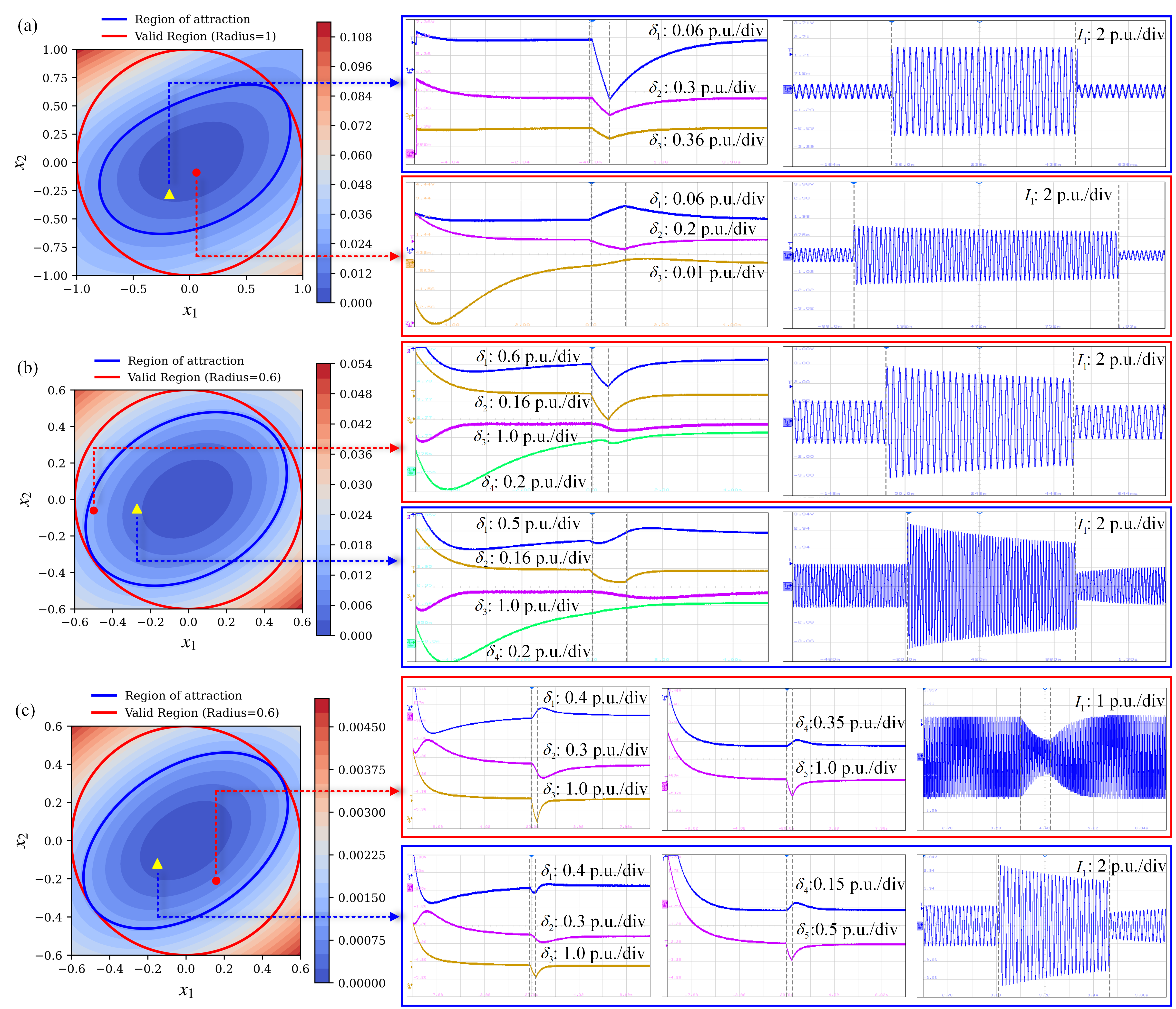}
    \vspace{-0.7cm}
    \caption{Time-domain verification of the estimated ROAs under large-signal disturbances. Responses of phase angles and output current of TS~\#1: (a) Case I  (b) Case II (c) Case III.}
    \label{waves}
    \vspace{-0.5cm}
\end{figure*}

Extensive time-domain HIL tests were conducted to validate these ROAs under large-signal disturbances, as visualized in Fig. \ref{waves}\, (b) and (c). 
For the 4-TS topology, two distinct fault scenarios were applied to evaluate the stability boundary.
First, concurrent catenary-to-rail short circuits with $R_f=0.08\,\Omega$ were applied to sections 1-2 and 3-4 during $t=5.0$--$5.5$~s. This disturbance displaced the post-fault system state vector to $[\delta_1, \delta_2, \delta_3, \delta_4] = [-0.27,\,-1.101,\,2.04,\,0.30]\,\text{rad}$, as denoted by the yellow triangle in Fig. \ref{waves}\, (b). Consistent with the prediction, the waveforms in the blue frame confirm that the system successfully withstands the disturbance, with phase angles and currents rapidly converging to their pre-fault equilibrium points after the fault clearance.
Another test was conducted by applying catenary-to-rail short circuits ($R_f=0\,\Omega$) to the same sections for a longer duration of $t=5.0$--$6.0$~s. This severe fault pushed the system state to $[\delta_1, \delta_2, \delta_3, \delta_4] = [-0.50,\,-1.11,\,2.21,\,0.42]\,\text{rad}$, marked by the red dot in Fig. \ref{waves}\, (b). The corresponding waveforms in the red frame show that the phase angles and output current of TS \#1 gradually recover to their pre-fault steady-state values, thus confirming the validity of estimated ROA.

Similarly, the transient stability of the 5-TS system was evaluated under single and multi-point fault conditions.
First, a catenary-to-rail short circuit with $R_f=0.08\,\Omega$ was applied to section 3-4 during $t=10.0$--$10.5$~s. The resulting perturbation shifted the state vector to $[\delta_1, \delta_2, \delta_3, \delta_4, \delta_5] = [0.16,\,-1.26,\,2.20,\,0.40,\,0.18]\,\text{rad}$, a critical operating point marked by the red dot in Fig. \ref{waves}\, (c). Despite the severity of this disturbance, the waveforms reveal that the system successfully retains synchronism, with the state trajectories smoothly converging back to the pre-fault equilibrium following fault clearance.
To further test the validity of the ROA, catenary-to-rail short circuits ($R_f=0\,\Omega$) were applied simultaneously to sections 1-2, 3-4, and 2-5 during $t=10.0$--$10.5$~s. This disturbance displaced the system state to $[\delta_1, \delta_2, \delta_3, \delta_4, \delta_5] = [-0.15,\,-1.17,\,2.56,\,0.8,\,0.02]\,\text{rad}$, corresponding to the yellow triangle in Fig. \ref{waves}\, (c). Consistent with the theoretical prediction, the associated waveforms demonstrate a rapid damping of oscillations, ultimately settling to a stable steady-state operating point.
In all investigated cases, the results confirm that the proposed framework provides a reliable and less conservative ROA, effectively handling the complexity of high-dimensional FTPSSs.

\section{Conclusion}
This paper presents a Lipschitz-guided learning framework designed to accelerate the transient stability verification of FTPSSs. By incorporating Lipschitz continuity into the verification phases, the proposed approach establishes a deterministic algebraic checking mechanism. This formulation effectively bypasses the heavy computational burden associated with exhaustive verification in high-dimensional state spaces, thereby substantially improving the computational efficiency of neural Lyapunov methods. The efficacy of this framework has been validated on FTPSSs of varying complexities. Comparative case studies demonstrate that the proposed method achieves a computational speedup of up to 12 times compared to state-of-the-art. Furthermore, this efficiency is achieved without compromising the accuracy of the stability boundary estimation. Unlike conventional analytical methods, such as LMI and SOS techniques, which often yield conservative results due to model approximations, our framework directly addresses the intrinsic sinusoidal nonlinearities of the network, providing significantly less conservative estimates of the ROA. Extensive case studies further confirm that the estimated stability boundaries are reliable under large-signal disturbances. Collectively, these results indicate that the proposed framework offers a tractable and efficient pathway for applying learning-based stability assessment to complex, high-order FTPSSs.

\appendices
\section{Autonomous System and Lipschitz Continuity}
\label{deflip}
Consider the autonomous nonlinear system described by the following ordinary differential equation:
\begin{equation}
    \dot{x}(t)=f\!\left(x(t)\right),\ x(0)=x_0,
    \label{system}
\end{equation}
where $x(t) \in \mathcal{D} \subseteq \mathbb{R}^n$ denotes the system state variable vector, $\mathcal{D}$ is the domain containing the origin, and $x_0\in\mathcal{D}$ is the initial condition. The function $f: \mathcal{D} \to \mathbb{R}^n$ is a nonlinear vector field that characterizes the system's evolution.

\begin{definition}[Lipschitz Continuity]
    The function $f$ is said to be Lipschitz continuous on $\mathcal{D}$ if there exists a positive constant $L \ge 0$, such that:
\begin{equation}
    \| f(x) - f(y) \| \le L \| x - y \|, \quad \forall x, y \in \mathcal{D}
\end{equation}
where $\|\cdot \|$ denotes a standard vector norm (typically the Euclidean norm), and the scalar $L$ is referred to as the Lipschitz constant. 
\end{definition} 
\section{Proof of Lipschitz Verification Guarantee}
\label{app:proof}

\begin{theorem}
Let $\mathcal{L}_f V_\theta(x)$ be the Lie derivative function defined over the domain $\Omega$. Let $\mathcal{K}_{comp} = \hat{L}_V (1 + \hat{L}_f)$ be the estimated Lipschitz constant bounding the variation of $\mathcal{L}_f V_\theta(x)$. Let $\mathcal{G} \subset \Omega$ be a set of mesh points with spatial resolution $h$, such that for any $x \in \Omega$, $\min_{x_g \in \mathcal{G}} \|x - x_g\|_2 \le h$.

If the condition:
\begin{equation}
    \mathcal{L}_f V_\theta(x_g) < -\mathcal{K}_{comp} h
\end{equation}
holds for all $x_g \in \mathcal{G}$, then $\mathcal{L}_f V_\theta(x) < 0$ for all $x \in \Omega$.
\end{theorem}

\begin{proof}
Consider an arbitrary point $x \in \Omega$ in the continuous domain. By the definition of the grid coverage, there must exist a nearest mesh point $x_g \in \mathcal{G}$ such that the Euclidean distance satisfies $\|x - x_g\|_2 \le h$.

By the definition of the Lipschitz constant $\mathcal{K}_{comp}$ for the composite function $\mathcal{L}_f V_\theta(x)$, the variation of the function value is bounded by:
\begin{equation}
     \left| \mathcal{L}_f V_\theta(x) - \mathcal{L}_f V_\theta(x_g) \right| \leq \mathcal{K}_{comp} \|x - x_g\|_2
\end{equation}

Rearranging the inequality to upper bound the value at the unmonitored point $x$:
\begin{equation}
    \mathcal{L}_f V_\theta(x) \leq \mathcal{L}_f V_\theta(x_g) + \mathcal{K}_{comp} \|x - x_g\|_2
\end{equation}

Substituting the verified condition $\mathcal{L}_f V_\theta(x_g) < -\mathcal{K}_{comp} h$ into the inequality:
\begin{equation}
    \mathcal{L}_f V_\theta(x) < -\mathcal{K}_{comp} h + \mathcal{K}_{comp} \|x - x_g\|_2
\end{equation}

Since $\|x - x_g\|_2 \le h$ and $\mathcal{K}_{comp} > 0$, the term $\mathcal{K}_{comp} \|x - x_g\|_2$ is at most $\mathcal{K}_{comp} h$. Combining this with the previous strict inequality yields:
\begin{equation}
    \mathcal{L}_f V_\theta(x) < -\mathcal{K}_{comp} h + \mathcal{K}_{comp} h = 0
\end{equation}

Thus, $\mathcal{L}_f V_\theta(x) < 0$ holds strictly for the entire continuous domain $\Omega$.
\end{proof}
\bibliographystyle{IEEEtran} 
\bibliography{ref} 

@ARTICLE{intro,
  author={Gu, Yunjie and Green, Timothy C.},
  journal={Proceedings of the IEEE}, 
  title={Power System Stability With a High Penetration of Inverter-Based Resources}, 
  year={2023},
  volume={111},
  number={7},
  pages={832-853},
  doi={10.1109/JPROC.2022.3179826}}

@ARTICLE{Tong,
  author={Huang, Tong and Gao, Sicun and Xie, Le},
  journal={IEEE Transactions on Smart Grid}, 
  title={A Neural Lyapunov Approach to Transient Stability Assessment of Power Electronics-Interfaced Networked Microgrids}, 
  year={2022},
  volume={13},
  number={1},
  pages={106-118},
  doi={10.1109/TSG.2021.3117889}}

@ARTICLE{SOS,
  author={Song, Qingchao and Chen, Jiawei and Loo, Ka-Hong and Chen, Jie and Chen, Pengwei},
  journal={IEEE Transactions on Power Electronics}, 
  title={Large-Signal Stability Analysis of Two-Stage Cascaded DC/DC Converter Systems Using Sum-of-Squares Programming}, 
  year={2024},
  volume={39},
  number={2},
  pages={2076-2085},
  doi={10.1109/TPEL.2023.3330841}}

@ARTICLE{LMI,
  author={Wang, Zhi and Guo, Li and Li, Xialin and Pang, Xiulan and Li, Xiaofeng and Zhou, Xu and Wang, Chengshan},
  journal={IEEE Transactions on Power Electronics}, 
  title={PLL Synchronization Transient Stability Analysis of a Weak-Grid Connected VSC During Asymmetric Faults}, 
  year={2024},
  volume={39},
  number={2},
  pages={2140-2154},
  doi={10.1109/TPEL.2023.3333808}}

@article{my,
  author={Liu, Zhong and Zheng, Jialin and Lu, Xiaonan},
  journal={IEEE Transactions on Power Electronics}, 
  title={Dissipation-Based Dynamics-Aware Learning Scheme for Transient Stability Analysis of Networked Black-Box Grid-Forming Inverters}, 
  year={2026},
  volume={41},
  number={3},
  pages={3165-3170},
}

@ARTICLE{SCU,
  author={Liu, Yunfei and Zhang, Junran and Liu, Yan and Yang, Mengling and Chen, Song and Zhou, Lijun and Wang, Yang},
  journal={IEEE Transactions on Smart Grid}, 
  title={An Improved Neural Lyapunov Method for Transient Stability Assessment of Networked Microgrids}, 
  year={2024},
  volume={15},
  number={2},
  pages={1410-1422},
  doi={10.1109/TSG.2023.3301855}}

@InProceedings{chuchu,
  title = 	 {Learning to stabilize high-dimensional unknown systems using {L}yapunov-guided exploration},
  author =       {Zhang, Songyuan and Fan, Chuchu},
  booktitle = 	 {Proceedings of the 6th Annual Learning for Dynamics \& Control Conference},
  pages = 	 {52--67},
  year = 	 {2024},
  volume = 	 {242},
  month = 	 {15--17 Jul},
  publisher =    {PMLR},
  url = 	 {https://proceedings.mlr.press/v242/zhang24a.html},
}

@ARTICLE{TDS,
  author={Fu, Xikun and Sun, Jianjun and Huang, Meng and Tian, Zhen and Yan, Han and Iu, Herbert Ho-Ching and Hu, Pan and Zha, Xiaoming},
  journal={IEEE Transactions on Power Electronics}, 
  title={Large-Signal Stability of Grid-Forming and Grid-Following Controls in Voltage Source Converter: A Comparative Study}, 
  year={2021},
  volume={36},
  number={7},
  pages={7832-7840},
  doi={10.1109/TPEL.2020.3047480}}

@InProceedings{sicun,
author="Gao, Sicun
and Avigad, Jeremy
and Clarke, Edmund M.",
editor="Gramlich, Bernhard
and Miller, Dale
and Sattler, Uli",
title="$\delta$-Complete Decision Procedures for Satisfiability over the Reals",
booktitle="Automated Reasoning",
year="2012",
publisher="Springer Berlin Heidelberg",
address="Berlin, Heidelberg",
pages="286--300",
isbn="978-3-642-31365-3"
}

@ARTICLE{theorem,
  author={Sauer, P. W. and LaGesse, D. J. and Ahmed-Zaid, S. and Pai, M. A.},
  journal={IEEE Transactions on Power Systems}, 
  title={Reduced Order Modeling of Interconnected Multimachine Power Systems Using Time-Scale Decomposition}, 
  year={1987},
  volume={2},
  number={2},
  pages={310-319},
  doi={10.1109/TPWRS.1987.4335124}}

@INPROCEEDINGS{ECCE,
  author={Liu, Zhong and Zheng, Jialin and Lu, Xiaonan},
  booktitle={2025 IEEE Energy Conversion Conference Congress and Exposition (ECCE)}, 
  title={Neural Lyapunov Based Transient Stability Analysis of Networked Grid-Forming Inverters with Unknown Internal Dynamics}, 
  year={2025},
  volume={},
  number={},
  pages={1-6},
  doi={10.1109/ECCE58356.2025.11259834}}

@book{nonlinear,
author = {Khalil, Hassan K.},
address = {Upper Saddle River, N.J.},
booktitle = {Nonlinear systems},
edition = {3rd edition., International ed.},
isbn = {9780131227408},
language = {eng},
publisher = {Pearson Education},
title = {Nonlinear Systems},
year = {2000},
}

@article{hu2024traction,
  author  = {Hu, Haitao and Liu, Yunjiang and Li, Yong and He, Zhengyou and Gao, Shibin and Zhu, Xiaojuan and Tao, Haidong},
  title   = {Traction power systems for electrified railways: evolution, state of the art, and future trends},
  journal = {Railway Engineering Science},
  volume  = {32},
  number  = {1},
  pages   = {1--19},
  year    = {2024},
  doi     = {10.1007/s40534-023-00320-6}
}

@article{li2025dispatch,
  author  = {Li, Ruofan and Sun, Qianhao and Chen, Qifang and Xia, Mingchao},
  title   = {A Methodology for Power Dispatch Based on Traction Substation Clusters in Flexible Traction Power Supply System},
  journal = {IEEE Transactions on Transportation Electrification},
  volume  = {11},
  number  = {4},
  pages   = {9976--9988},
  month   = aug,
  year    = {2025},
  doi     = {10.1109/tte.2025.3565948}
}

@article{guo2024bipolar,
  author  = {Guo, Peiqian and Tian, Zhongbei and Yuan, Zhichang and Zhang, Xiao-Yu and Sharifi, Dela},
  title   = {Research on Symmetric Bipolar {MMC-M}$^{2}${Tdc}-Based Flexible Railway Traction Power Supply System},
  journal = {IEEE Transactions on Transportation Electrification},
  volume  = {10},
  number  = {1},
  pages   = {1043--1055},
  month   = mar,
  year    = {2024},
  doi     = {10.1109/tte.2023.3290591}
}

@article{li2022adrc,
  author  = {Li, Ziming and Li, Xiaoqian and Lin, Yunzhi and Wei, Yingdong and Li, Zhanhe and Li, Zengqin and Lu, Chao},
  title   = {Active Disturbance Rejection Control for Static Power Converters in Flexible AC Traction Power Supply Systems},
  journal = {IEEE Transactions on Energy Conversion},
  volume  = {37},
  number  = {4},
  pages   = {2851--2862},
  month   = dec,
  year    = {2022},
  doi     = {10.1109/tec.2022.3187029}
}

@article{li2023spc,
  author  = {Li, Jing and Wei, Yingdong and Li, Xiaoqian and Lu, Chao and Guo, Xu and Lin, Yunzhi and Molinas, Marta},
  title   = {Modeling and stability prediction for the static-power-converters interfaced flexible {AC} traction power supply system with power sharing scheme},
  journal = {International Journal of Electrical Power \& Energy Systems},
  volume  = {154},
  pages   = {109401},
  month   = dec,
  year    = {2023},
  doi     = {10.1016/j.ijepes.2023.109401}
}

@article{guo2025multitimescale,
  author  = {Guo, Qi and Wang, Xin and Tu, Chunming and Che, Liang and Xiao, Fan and Hou, Yuchao},
  title   = {A Multitimescale Energy Management Strategy for Advanced Cophase Traction Power Supply System With {PV} and {ES}},
  journal = {IEEE Transactions on Transportation Electrification},
  volume  = {11},
  number  = {1},
  pages   = {859--869},
  month   = feb,
  year    = {2025},
  doi     = {10.1109/tte.2024.3396878}
}

@article{chen2026flexible,
  author  = {Chen, Ziyi and Wang, Xin and Guo, Qi and Tu, Chunming and Xiao, Fan and Li, Junhao},
  title   = {Flexible Reconfiguration Strategy for Novel Advanced Traction Power Supply Equipment Under Fault Conditions},
  journal = {IEEE Transactions on Transportation Electrification},
  volume  = {12},
  number  = {4},
  pages   = {5893--5902},
  month   = aug,
  year    = {2026},
  doi     = {10.1109/tte.2026.3675516}
}

@article{meng2025neural,
  author  = {Meng, Xiangyu and Hu, Guiyang and Liu, Zhigang and Wang, Hui and Zhang, Guinan and Lin, Hongjian and Sadabadi, Mahdieh S.},
  title   = {Neural Network-Based Impedance Identification and Stability Analysis for Double-Sided Feeding Railway Systems},
  journal = {IEEE Transactions on Transportation Electrification},
  volume  = {11},
  number  = {1},
  pages   = {4429--4442},
  month   = feb,
  year    = {2025},
  doi     = {10.1109/tte.2024.3462857}
}

@article{liu2024impedance,
  author  = {Liu, Feifan and Ma, Lei and Zhou, Fulin and Cheng, Shaokun and Xiong, Jinfei},
  title   = {Impedance Modeling and Stability Analysis in Co-Phase Traction Power Supply System},
  journal = {IEEE Transactions on Transportation Electrification},
  volume  = {10},
  number  = {4},
  pages   = {9712--9726},
  month   = dec,
  year    = {2024},
  doi     = {10.1109/tte.2024.3355141}
}

@article{chen2026frequency,
  author  = {Chen, Minwu and Fu, Jicheng and Kong, Rui and Chen, Yinyu and Blaabjerg, Frede},
  title   = {Modeling and Stability Analysis for Frequency-Coupled Oscillation in Train-Network Interactions},
  journal = {IEEE Transactions on Transportation Electrification},
  note    = {early access},
  year    = {2026},
  doi     = {10.1109/tte.2026.3715644}
}

@article{zheng2025shipboard,
  author  = {Zheng, Litao and Fang, Sidun and Luo, Yingbing and Niu, Tao and Chen, Guanhong and Liao, Ruijin},
  title   = {Large-Signal Stability Analysis and Control of {DC} Shipboard Microgrids Considering Battery Dynamic Characteristics},
  journal = {IEEE Transactions on Transportation Electrification},
  volume  = {11},
  number  = {5},
  pages   = {11940--11954},
  month   = oct,
  year    = {2025},
  doi     = {10.1109/tte.2025.3584110}
}

@article{xu2025dab,
  author  = {Xu, Peichao and Wen, Huiqing and Bu, Qinglei and Dai, Lanbo and Zhu, Yi and Khan, Shahamat Shahzad and Shi, Haochen and Zhou, Jiafeng and Ma, Jieming and Hu, Yihua and Rodriguez, Jose},
  title   = {Large-Signal Stability for {DAB} Converters Feeding Constant Power Load Based on State Trajectory Analysis},
  journal = {IEEE Transactions on Transportation Electrification},
  volume  = {11},
  number  = {1},
  pages   = {3244--3261},
  month   = feb,
  year    = {2025},
  doi     = {10.1109/tte.2024.3436635}
}

@article{chiang1989existence,
  author  = {Chiang, H.-D.},
  title   = {Study of the existence of energy functions for power systems with losses},
  journal = {IEEE Transactions on Circuits and Systems},
  volume  = {36},
  number  = {11},
  pages   = {1423--1429},
  month   = nov,
  year    = {1989},
  doi     = {10.1109/31.41298}
}

@article{xing2023thermal,
  author = {Xing, Chen and Li, Kang and Su, Jialei},
  title = {Thermal Constrained Energy Optimization of Railway Cophase Systems With {ESS} Integration---An {FRA}-Pruned {DQN} Approach},
  journal = {IEEE Transactions on Transportation Electrification},
  volume = {9},
  number = {4},
  pages = {5122--5139},
  year = {2023},
  doi = {10.1109/TTE.2022.3218762}
}

@article{zhang2023thyristor,
  author = {Zhang, Zhi and Li, Kai and Zhang, Zhibo and Zheng, Trillion Q. and Hao, Ruixiang and You, Xiaojie},
  title = {A Thyristor-Based Auto-Passing Neutral Section Scheme With Auxiliary Transformers for {AC} Electrified Railway},
  journal = {IEEE Transactions on Transportation Electrification},
  volume = {9},
  number = {2},
  pages = {2296--2307},
  year = {2023},
  doi = {10.1109/TTE.2022.3207212}
}

@article{ge2024wide,
  author = {Ge, Yinbo and Huang, Yi and Chen, Junyu and Wang, Ke and Hu, Haitao and He, Zhengyou},
  title = {Wide-Adaptable Uninterrupted Power Supply System for Neutral Sections of {ERs} With Event-Triggered Flexible Power and Voltage Control},
  journal = {IEEE Transactions on Power Electronics},
  volume = {39},
  number = {12},
  pages = {16566--16579},
  year = {2024},
  doi = {10.1109/TPEL.2024.3444749}
}

@article{chen2024power,
  author = {Chen, Junyu and Zhao, Yue and Wang, Minghao and Wang, Ke and Huang, Yi and Xu, Zhao},
  title = {Power Sharing and Storage-Based Regenerative Braking Energy Utilization for Sectioning Post in Electrified Railways},
  journal = {IEEE Transactions on Transportation Electrification},
  volume = {10},
  number = {2},
  pages = {2677--2688},
  year = {2024},
  doi = {10.1109/TTE.2023.3295089}
}

@article{luo2025multiagent,
  author = {Luo, Jiaming and Wang, Yating and Gao, Shibin and Wei, Xiaoguang and Tian, Zhongbei and Shi, Jian},
  title = {Multiagent-Game-Based Reinforcement Learning Energy Management Strategy for Flexible Traction Power Supply System With Energy Storage System},
  journal = {IEEE Transactions on Transportation Electrification},
  volume = {11},
  number = {3},
  pages = {8474--8488},
  year = {2025},
  doi = {10.1109/TTE.2025.3541556}
}

@article{chen2022optimal,
  author = {Chen, Minwu and Liang, Zongyou and Cheng, Zhe and Zhao, Jinyu and Tian, Zhongbei},
  title = {Optimal Scheduling of {FTPSS} With {PV} and {HESS} Considering the Online Degradation of Battery Capacity},
  journal = {IEEE Transactions on Transportation Electrification},
  volume = {8},
  number = {1},
  pages = {936--947},
  year = {2022},
  doi = {10.1109/TTE.2021.3093321}
}

@article{ge2026series,
  author = {Ge, Yinbo and Li, Zilin and Chen, Junyu and Lin, Hongjian and Wang, Ke and Hu, Haitao and He, Zhengyou},
  title = {Series Hybrid Flexible Traction Substation for {AC} Electric Railways With Reduced Converter Capacity},
  journal = {IEEE Transactions on Industrial Electronics},
  volume = {73},
  number = {6},
  pages = {8455--8467},
  year = {2026},
  doi = {10.1109/TIE.2025.3647927}
}

@article{ying2023online,
  author = {Ying, Yichen and Liu, Qiujiang and Wu, Mingli and Zhai, Yating},
  title = {Online Energy Management Strategy of the Flexible Smart Traction Power Supply System},
  journal = {IEEE Transactions on Transportation Electrification},
  volume = {9},
  number = {1},
  pages = {981--994},
  year = {2023},
  doi = {10.1109/TTE.2022.3192141}
}

@article{hu2024decentralized,
  author = {Hu, Haitao and Liu, Yunjiang and Yang, Xiaowei and Li, Yong and Li, Zhaoyang and Zhu, Xiaojuan},
  title = {A Decentralized Coordinated Control Strategy for Multiple Traction Substations in {MVdc} Railway Power Supply System},
  journal = {IEEE Transactions on Transportation Electrification},
  volume = {10},
  number = {2},
  pages = {3246--3257},
  year = {2024},
  doi = {10.1109/TTE.2023.3291926}
}

@article{wang2024emergency,
  author = {Wang, Xin and Guo, Qi and Tu, Chunming and Wang, Lei and Xiao, Fan and Hou, Yuchao},
  title = {Emergency Control Strategy of Hybrid Advanced Traction Power Supply System Based on Multidevice Collaboration},
  journal = {IEEE Transactions on Transportation Electrification},
  volume = {10},
  number = {2},
  pages = {2367--2376},
  year = {2024},
  doi = {10.1109/TTE.2023.3297316}
}
\end{document}